\documentclass[journal,comsoc]{IEEEtran}
\usepackage[T1]{fontenc}% optional T1 font encoding
\usepackage{graphicx}

\graphicspath{ {images/} }
\usepackage{algorithmic}
\usepackage{algorithm}
\usepackage{amsmath}
\usepackage{amssymb}
\usepackage{textcomp}
\usepackage{amsthm}
\usepackage{wrapfig}
\usepackage{lipsum}
\usepackage{xcolor}
	
\usepackage{soul}

\graphicspath{ {images/} }
\usepackage{array}
\ifCLASSOPTIONcompsoc
  \usepackage[caption=false,font=normalsize,labelfont=sf,textfont=sf]{subfig}
\else
  \usepackage[caption=false,font=footnotesize]{subfig}
\fi
\usepackage{stfloats}
\usepackage{url}
\usepackage{balance}
\newtheorem{theorem}{Theorem}
\newtheorem{corollary}{Corollary}[theorem]
\newcommand{\boldm}[1] {\mathversion{bold}#1\mathversion{normal}}
\usepackage[cmintegrals]{newtxmath}

\begin{document}

\title{Wireless Channel Modeling Based on Extreme Value Theory for Ultra-Reliable Communications}

\author{Niloofar~Mehrnia,~\IEEEmembership{Student Member,~IEEE,}
        Sinem~Coleri,~\IEEEmembership{Senior Member,~IEEE}% <-this % stops a space
\thanks{Niloofar Mehrnia and Sinem Coleri are with the Department of Electrical and Electronics Engineering,
Koc University, Istanbul, Turkey (e-mail: nmehrnia17@ku.edu.tr; scoleri@ku.edu.tr).}% <-this % stops a space
\thanks{Niloofar Mehrnia is also with Koc University Ford Otosan Automotive Technologies Laboratory (KUFOTAL), Sariyer, Istanbul, Turkey, 34450.}
\thanks{Sinem Coleri acknowledges the support of Ford Otosan.}
}% <-this % stops a space

%{Shell \MakeLowercase{\textit{et al.}}: Bare Demo of IEEEtran.cls for IEEE Communications Society Journals}

\maketitle

\begin{abstract}
A key building block in the design of ultra-reliable communication systems is a wireless channel model that captures the statistics of rare events occurring due to the significant fading. In this paper, we propose a novel methodology based on extreme value theory (EVT) to statistically model the behavior of extreme events in a wireless channel for ultra-reliable communication. This methodology includes techniques for fitting the lower tail distribution of the received power to the generalized Pareto distribution (GPD), determining the optimum threshold over which the tail statistics are derived, ascertaining the optimum stopping condition on the number of samples required to estimate the tail statistics by using GPD, and finally, assessing the validity of the derived Pareto model.
Based on the data collected within the engine compartment of Fiat Linea under various engine vibrations and driving scenarios, we demonstrate that the proposed methodology provides the best fit to the collected data, significantly outperforming the conventional extrapolation-based methods. Moreover, the usage of the EVT in the proposed method decreases the required number of samples for estimating the tail statistics significantly.  
\end{abstract}
%\vspace{0.2cm}

\IEEEpeerreviewmaketitle
%\vspace{0.3cm}
\section{Introduction}
\label{sec:intro}
%\vspace{0.2cm}
\IEEEPARstart{U}{ltra}-reliable and low-latency communication (URLLC) is an essential part of beyond $5$th generation ($5$G) networks with the potential to enable mission-critical applications, such as remote control of robots, self-driving cars, and remote surgeries  \cite{interface_01}-\nocite{5G_01}\nocite{interface_02}\cite{urllc_02}. URLLC is defined as a communication with targeted packet error rate in the range of $10^{-9}-10^{-5}$, and acceptable latency on the order of milliseconds or less \cite{urllc_02}-\cite{interface_03}.
Addressing the strict requirements of URLLC necessitates fundamental breakthroughs in the statistical modeling of the wireless channel in the ultra-reliable region, incorporating novel techniques to analyze the lower tail of the distributions encompassing extremely low probabilities, and handling and optimizing a large amount of data to model the extreme events occurring infrequently. A proper channel modeling approach is the key to achieve URLLC not only at the physical layer but also at the upper layers, including data link and network layers.

Previous studies on the channel modeling for URLLC focus on either providing a unified framework by extrapolating a wide range of practically important channel models to extend their applicability to the ultra-reliability regime of operation \cite{urllc_02}, or proposing new channel parameters incorporating extreme reliability requirements into the communication \cite{urllc_07}\nocite{urllc_08}-\cite{urllc_05}. In \cite{urllc_02}, a simple power-law expression is proposed for estimating the tail of the cumulative distribution function (CDF) of the received power in block fading channels in the regime of extremely rare events by extrapolating the commonly used practical fading models. On the other hand, the authors in \cite{urllc_07} and \cite{urllc_08} propose a new channel parameter by challenging the definition of coherence time, during which channels are considered to be static for an average performance of traditional cellular and WiFi networks. As an alternative, a more nuanced notion of coherence time, considering an ultra-reliability regime, is defined as the time over which a channel is predictable to a given reliability. In \cite{urllc_05}, two substitute performance measures have been provided for the reliability of the channel, considering that the exact knowledge of the channel is not available at the ultra-reliability level: averaged reliability, for dynamic environments in which the channel changes frequently; and probably correct reliability, for the static environments in which the most recent channel estimate can be used by the system over the future transmissions. However, still no wireless channel modeling framework has been proposed to derive and verify these ultra-reliability statistics. Deriving an appropriate wireless channel model is immensely important as model uncertainty and/or mismatch degrades the communication performance by several orders of magnitude, which is not acceptable at URLLC level \cite{urllc_05}. One might think that deriving the tail statistics is equivalent to collecting a large amount of data and fitting these data to probabilistic distributions. However, these derived distributions may not capture rare events due to the limited amount of collected data or may not be valid in another setting with different environmental conditions, such as different temperatures and vibrations. Extreme value theory (EVT) is a unique statistical discipline that develops techniques and models for describing rare events based on the implementation of mathematical limits as finite-level approximation \cite{urllc_02}, \cite{evt_03}-\cite{evt_04}.

EVT has been utilized at data link and network layers to model the tail statistics of queue length and delay \cite{upperlayer_03}\nocite{upperlayer_05}\nocite{bennisAoI}\nocite{upperlayer_04}\nocite{upperlayer_02}-\cite{upperlayer_01}, in a limited context at physical layer for wireless channel modeling, to provide a fit to the whole distribution of large-scale or small-scale fading \cite{fadingevt_01}\nocite{fitevt_01}\nocite{fitevt_02}\nocite{fitevt_03}-\cite{fitevt_04}, and to analyze the asymptotic behavior of the ergodic secrecy rate (ESR)/ergodic multicast rate \cite{Xu2016}-\cite{Subhash2020}.
EVT has been applied in \cite{upperlayer_03} and \cite{upperlayer_05} to characterize the statistics of the large queue lengths and accordingly propose a queue-aware resource allocation approach for enabling ultra reliable and low latency vehicular communications. The authors in \cite{bennisAoI} adopt EVT to characterize the tail, and the excess value of the vehicles’ queueing systems and arrival rates to further incorporate its results in the transmission power minimization problem. 
Moreover, EVT is utilized in \cite{upperlayer_04}\nocite{upperlayer_02}-\cite{upperlayer_01} to derive closed-form asymptotic expressions for the throughput, bit error rate (BER), and packet error rate (PER) over different fading channels. In \cite{upperlayer_04}, EVT is used to find a limiting distribution for the PER in the additive white Gaussian noise (AWGN) channel with the goal of indicating the superiority of EVT-based approach to accurately approximate the average PER for the uncoded schemes. 
In \cite{upperlayer_02}, the authors apply EVT to find the limiting distribution of the maximum throughput and then, alleviate the complexity of the average throughput analysis for the channels in multi-user diversity. Likewise, in \cite{upperlayer_01}, EVT is used to analyze the effective throughput, average throughput and average bit error probability of the $k$-th best link with the highest throughput in the selection diversity schemes over various fading channels, such as Weibull, Gamma, $\alpha - \mu$ and Gamma-Gamma. 
Nevertheless, these data link and network layer studies use the existing average statistics-based channel models, and therefore, their extrapolation in the ultra-reliable region, the accuracy of which has not yet been analyzed experimentally. 
On the other hand, the authors in \cite{fadingevt_01}-\cite{fitevt_01}, \cite{fitevt_03} fit extreme value distribution (EVD) to the whole path-loss or power distribution, concluding that EVD models wireless channel fading better than the well-known models, such as Rayleigh and Rician. Additionally, in \cite{fitevt_02} and \cite{fitevt_04}, generalized extreme value (GEV) distribution is used to model the small scale fading in maritime communication, and root-mean-square (RMS) delay spread for vehicle-to-vehicle (V2V) communication, respectively. Moreover, EVT is applied in \cite{Xu2016} to fit Gumbel distribution to the lower bound of the ESR, which is then incorporated as the statistical constraints within privacy-preserving problem and achieves higher secrecy rate for downlink transmission in multiuser relay systems. Also, in \cite{Subhash2020} and \cite{secrecy_01}, EVT is used to fit a distribution to the minimum of signal-to-interference ratio (SIR) random variables and then  determine the optimum value of the transmit power of a secondary network in a Cognitive Radio Network (CRN) subject to Quality of Service (QoS) constraints. However, EVT has never been incorporated into wireless channel modeling to estimate the tail statistics nor to use the theorems for the consistency of the distributions, stopping conditions on determining the sufficient number of samples required in EVT, and validation procedures to address the reliability constraint at the URLLC levels. 

The goal of this paper is to propose a novel channel modeling methodology for URLLC based on EVT to derive the lower tail statistics of the received signal power in a consistent manner while efficiently dealing with a massive amount of corresponding data, for the first time in the literature. We use received signal power at constant transmit power as channel data as it is equivalent to the squared amplitude of the channel state information \cite{urllc_02}. The modeling approach adapts EVT to (i) determine the optimum threshold over which the tail statistics are derived based on the assumption that all values exceeding the threshold correspond to the extreme events happening rarely, (ii) determine the stopping conditions for ascertaining the minimum amount of data required to model the tail characteristics, and (iii) assess the validity of the final model by using probability plots. The original contributions of the paper are listed as follows:

\begin{itemize}

    \item We propose a comprehensive channel modeling approach for URLLC based on EVT, for the first time in the literature. The methodology includes techniques for fitting the tail distribution to the generalized Pareto distribution (GPD), determining the optimum threshold over which consistent tail statistics are derived, specifying the optimum stopping conditions on the sufficient number of measurement samples, and assessing the model validity.
    
    \item We derive the parameters of the tail distribution of the received power by fitting the GPD to the independent and identically distributed (i.i.d.) samples extracted from two methods: Auto-Regressive Integrated Moving Average (ARIMA)-Generalized Auto-Regressive Conditional Heteroskedasticity (GARCH) and declustering method. Both ARIMA-GARCH and declustering methods remove the dependency among the observations and provide i.i.d. samples to the EVT input.
    
    \item We propose a novel methodology for determining the optimum threshold over which the tail statistics are derived, for the first time in the wireless communication literature. The existing studies in the URLLC regime of operation assume that the value of the threshold in EVT is known \cite{upperlayer_03}\nocite{upperlayer_05}\nocite{bennisAoI}\nocite{upperlayer_04}\nocite{upperlayer_02}-\cite{upperlayer_01}, while the choice of optimum threshold is of paramount importance as it separates non-extreme observations from the extreme ones. We apply two complementary methods adopted from EVT to improve the accuracy of the threshold estimation: mean residual life method that is applied prior to the estimation of the tail distribution by using GPD, and parameter stability method that is applied after deriving the tail statistics by using GPD. 

    \item We propose a novel methodology for determining the stopping conditions on the number of measured samples, sufficiently large to model the channel tail statistics in a consistent manner based on the adaption of the framework in \cite{evt_09}, for the first time in the literature. The adaptation is made to the minimum number of samples required to estimate the distribution of values exceeding a given threshold, instead of the probability distribution of the maxima/minima in a time series in \cite{evt_09}, and using GPD, rather than GEV distribution in \cite{evt_09}.

    \item We propose a validation procedure for the accuracy assessment of the channel tail model derived by EVT using the probability plots, for the first time in the literature. We utilize the probability-probability (PP) plot and quantile-quantile (QQ) plot to assess the goodness-of-fit (GOF) of the Pareto model that characterizes the channel tail distribution.
    
    \item We demonstrate the superiority of the proposed methodology for deriving the tail statistics in terms of the modeling accuracy, compared to the conventional method based on the extrapolation of the average statistics to the ultra-reliable region, over the data collected within the engine compartment of Fiat Linea under various engine vibrations and driving scenarios.
   
\end{itemize}

The rest of the paper is organized as follows. Section \ref{sec:background} describes the basics of EVT together with the theorems and corollaries used in the development of the proposed channel modeling approach, and the methods for removing the dependency in the observation samples and providing the i.i.d. input to the EVT. Section~\ref{sec:channel} presents the proposed channel modeling framework based on EVT for characterizing the channel tail distribution in the ultra-reliable region. Section \ref{sec:performance evaluation} provides the channel measurement setup, and the performance evaluation of the proposed algorithm in determining the optimum threshold and minimum number of required samples for the derivation of the channel tail statistics, and compared to the conventional methods in terms of the estimation accuracy. Finally, concluding remarks and future works are given in Section \ref{sec:conclusions}.

%\vspace{0.3cm}
\section{Background}
\label{sec:background}
%\vspace{0.2cm}
\subsection{Extreme Value Theory}
\label{sec:evt}
EVT is a powerful tool for characterizing the probabilistic distribution of extreme events occurring with low probability. EVT has been used for a long time in a wide variety of fields, such as hydrology to quantify risks of extreme floods, rainfalls and waves \cite{evt_03}, or financial engineering to estimate the unexpected losses due to the extreme events \cite{evt_01}. However, it has only been recently employed in communication engineering, mainly in the analysis of the extreme values in network traffic, worst-case delay, queue lengths, throughput and BER/PER to address the strict delay and reliability requirements of URLLC \cite{upperlayer_03}, \cite{urllc_01}.

In general, one can divide EVT applications into two categories \cite{evt_01}\nocite{evt_02}\nocite{evt_03}\nocite{evt_04}\nocite{evt_05}\nocite{evt_06}\nocite{evt_07}-\cite{evt_08}. In the first category, the asymptotic distribution of the maxima/minima of a long finite sequence of i.i.d. random variables is modeled by the generalized extreme value (GEV) distribution, while in the second category, the distribution of the values exceeding a given threshold is modeled by the GPD \cite{evt_05}-\cite{evt_06}. 
In URLLC, the main focus is on the second category of EVT applications, as all values exceeding a given threshold are considered as extreme events and need to be incorporated into the tail statistic analysis. In the following, we briefly introduce EVT and its major results used to develop the methodologies in the upcoming sections.
  
\begin{theorem}
\label{theorem:evt}
Let $X = [ X_1, X_2, ..., X_n ]$ be a sequence of independent random variables with a CDF denoted by $F$. Then, for low enough threshold $u$, the probabilistic distribution of the values exceeding $u$, i.e., $Pr\{ u-X < y | X < u \}$, is approximated by the GPD with CDF given by
\begin{equation}
\label{eqn:gpd dist}
    H_{u}(y) = 1-\Big[1+\frac{\xi y}{\tilde{\sigma}_{u}}\Big]^{-1/\xi}.
\end{equation}
defined on $\{ y: y > 0\}$, where $\xi$ and $\tilde{\sigma}_{u}=\sigma+\xi(u-\mu)$ are shape and scale parameters of the GPD, respectively \cite{evt_11}. Here, $\mu$ and $\sigma$ are the location and scale parameters of GEV distribution fitted to the CDF of $M_n = max \{ X_1,...,X_n \}$, respectively \cite{evt_01}, \cite{evt_04}.
\end{theorem}

\begin{proof}
For large sample size $n$, the CDF of $M_{n}$ is given by
\begin{center}
$Pr\{ M_n \leq z \}
\approx F^{n}{(z)} $,\\
\end{center}
where
\begin{equation}
\label{eqn:evt dist}
    F^{n}(z) = \exp{\Big\{-\Big[1+\xi\Big(\frac{z-\mu}{\sigma}\Big)\Big]^{-1/\xi}\Big\}}
\end{equation}
is in the form of GEV distribution, and $\mu$, $\xi$, and $\sigma$ are the location, shape, and scale parameters of the GEV distribution, respectively \cite[Theorem~3.3]{evt_04}. 
By taking the logarithm of both sides of Eqn.~(\ref{eqn:evt dist}), and using the Taylor series expansion, we obtain

\begin{equation}
    \label{eqn:f(z)}
    1 - F(z) = n^{-1} \Big[ 1+\xi \Big( \frac{z-\mu}{\sigma} \Big) \Big]^{-1/\xi}.
\end{equation}
The CDF of the excesses over threshold $z$ is then given by
\begin{equation}
\label{eqn:probdif}
    \begin{split}
        Pr \{ X-z<y | X>z \}
        %=& 1 - Pr \{ X<u-y | X<u \} \\
        = 1- \frac{1-F(z+y)}{1-F(z)}\\
        %=& 1- \Bigg[ \frac{1+\xi (\frac{u-y-\mu}{\sigma})}{1+\xi(\frac{u-\mu}{\sigma})} \Bigg]^{-1/\xi}\\
        = 1- \Big[ 1+\frac{\xi y}{\sigma+\xi(z-\mu)} \Big]^{-1/\xi}.
    \end{split}
\end{equation}
Assuming that $\tilde{X}=-X$, the lower tail is transformed to the upper tail and then, GPD is fitted to $\{u-\tilde{X}|\tilde{X}<u\}$. With this transformation, the CDF of the excesses below threshold $u$ is expressed as $1- \Big[ 1+\frac{\xi y}{\tilde{\sigma}} \Big]^{-1/\xi}$, where $\tilde{\sigma}=\sigma-\xi(u+\mu)$.

%where 
%\begin{equation} 
%    \label{eqn:sigmatilde}
%    \tilde{\sigma}_{u} = \sigma + \xi (u-\mu).
%\end{equation}

\end{proof}
Theorem \ref{theorem:evt} implies that the distribution of the i.i.d. sequence of random variables exceeding a given threshold can be approximated by GPD. The parameters of the GPD are estimated by using the maximum likelihood estimator (MLE) \cite{evt_01}, \cite{evt_04}, \cite{MLE}. The consequences of Theorem \ref{theorem:evt} are given in the following corollaries. Corollaries \ref{cor:stability} and \ref{cor:meanexcces} will be used in Section \ref{sec:threshold} to determine the optimum threshold for estimating the channel tail distribution by using GPD, whereas Corollary \ref{cor:returnlevel} will be utilized in Section \ref{sec:sample size} to evaluate the minimum number of samples, sufficiently large to estimate the channel tail distribution by GPD.

\begin{corollary}\label{cor:stability}
Let $X = [ X_1, X_2, ..., X_n ]$ be a sequence of i.i.d. random variables, where GPD($\xi$,$\tilde\sigma_{u_{0}}$) models the probabilistic distribution of exceedances below threshold $u_{0}$.
Then, for all thresholds $u<u_{0}$, the CDF of values below threshold $u$ is modeled by GPD($\xi$,$\tilde\sigma_{u}$), where $\tilde \sigma_{u} = \tilde \sigma_{u_{0}} - \xi (u-u_{0})$, and the estimates of shape parameter ($\xi$) and modified scale parameter defined as $\sigma^{*}=\tilde{\sigma}_{u}+\xi u$ are constant against $u$. However, due to the sampling variability, the estimated $\xi$ and $\sigma^{*}$ cannot be exactly constant, but linear with respect to $u$.
\end{corollary}

\begin{proof}
By Theorem \ref{theorem:evt}, if a generalized Pareto distribution is a valid model for the excesses below threshold $u_{0}$, then the excesses of a threshold $u$, $u<u_{0}$, should also follow a generalized Pareto distribution. Accordingly, if exceedances below threshold $u_{0}$ have GPD($\xi$,$\tilde\sigma_{u_{0}}$), with $\tilde \sigma_{u_{0}} = \sigma - \xi (u_{0} + \mu)$, then, the exceedances below threshold $u<u_{0}$ are distributed by GPD($\xi$,$\tilde\sigma_{u}$), where $\tilde \sigma_{u} = \sigma - \xi (u + \mu)$. Therefore, $\tilde \sigma_{u} = \tilde \sigma_{u_{0}} - \xi (u-u_{0})$.
%\begin{equation}
%\label{eqn:stability}
%    \tilde \sigma_{u} = \tilde \sigma_{u_{0}} + \xi (u-u_{0}).
%\end{equation}
Though the shape parameter is threshold invariant, the scale parameter changes with $u$ unless $\xi = 0$. To make $\tilde \sigma_{u}$ constant with respect to $u$, we define the modified scale parameter as
%\vspace{-0.7cm}

\begin{equation}
\label{eqn:modscaleparam}
    %\begin{split}
        \sigma^{*} = \tilde\sigma_{u} + \xi u,
    %\end{split}
\end{equation}
where $\sigma^{*}$ depends only on $\tilde\sigma_{u_{0}}$ and $\xi$, so, is $u$ independent. It should be noted that the estimated $\xi$ and $\sigma^{*}$ cannot be exactly constant since different thresholds are associated with different number of samples included in the tail analysis. The MLEs of the estimated shape parameters change linearly with the number of samples included in the tail \cite{evt_04}, \cite{evt_02}, therefore, the shape parameters should be linear with respect to $u<u_{0}$. 
\end{proof}

\begin{corollary}\label{cor:meanexcces}
Let $X = [ X_1, X_2, ..., X_n ]$ be a sequence of i.i.d. random variables, where GPD($\xi$,$\tilde\sigma_{u_{0}}$) models the probabilistic distribution of exceedances below threshold $u_{0}$. Then, for all thresholds $u<u_{0}$, the mean excess function (MEF), also known as the mean residual life function, defined by $e(u) = E(u-X|X<u)$, is a linear function of $u$.
\end{corollary}

\begin{proof}
Let $Y$ be a random variable with a probabilistic distribution GPD($\xi$,$\sigma$). Then, the expectation of $Y$ is given by $E(Y) = {\sigma}/{(1-\xi)}$, where $E(.)$ denotes the expectation function. 
By Theorem \ref{theorem:evt}, and Corollary~\ref{cor:stability}, if GPD($\xi$,$\tilde\sigma_{u_{0}}$) models the probabilistic distribution of exceedances below threshold $u_{0}$, then, for any threshold $u$, $u<u_{0}$, the CDF of values below threshold $u$ is modeled by GPD($\xi$,$\tilde\sigma_{u_{0}} - \xi (u-u_{0})$) with the mean excess function 
\begin{equation}
    \begin{split}
        e(u) = E(u-X|X<u) 
        = \frac{\tilde\sigma_{u_{0}} - \xi (u-u_{0})}{1-\xi},
    \end{split}
\end{equation}
which is a linear function of $u$.
\end{proof}
According to Corollary~\ref{cor:meanexcces}, if the mean excesses are plotted against the threshold $u$, and the GPD assumption is satisfied at threshold $u_{0}$, then, for all thresholds $u<u_{0}$, the plot should be a straight line with slope and intercept equal to $\frac{-}\xi{1-\xi}$ and $\frac{\tilde\sigma_{u_{0}} + \xi u_{0}}{1-\xi}$, respectively.

\begin{corollary}\label{cor:returnlevel}
Let $X = [ X_1, X_2, ..., X_n ]$ be a sequence of i.i.d. random variables, where GPD($\xi$,$\tilde\sigma_{u}$) models the probabilistic distribution of exceedances below threshold $u$. Then, the return level $r_{m}<u$, defined as the extreme quantile that is expected to be exceeded with probability $1/m$ on average, is given by  
\begin{equation}
    \label{eqn:returnlevel}
    r_{m} = u + \frac{\tilde\sigma_{u}}{\xi} [(m \zeta_{u})^{\xi} -1],
\end{equation}
where $\zeta_{u} = Pr\{ X<u \}$, and $r_m$ is normally distributed if  the sample size $n$ is sufficiently large to estimate the tail distribution of values exceeding threshold $u$ by using GPD($\xi$,$\tilde \sigma_{u}$).
\end{corollary}
%\vspace{-0.45cm}
\begin{proof}
Assume $r=u-y$ in Eqn.~(\ref{eqn:gpd dist}). Then,
\begin{equation}
    Pr\{ X<r | X<u \} = \Big[ 1+\frac{\xi (u-r)}{\tilde\sigma_{u}} \Big]^{-1/\xi}.
\end{equation}

Since the return level, $r_{m}$ is the extreme quantile that is expected to be exceeded on average once in $m$ observations,
\begin{equation}
\label{eqn:retlev}
    \frac{1}{m} = \zeta_{u} \Big[ 1+\frac{\xi (u-r_{m})}{\tilde\sigma_{u}} \Big]^{-1/\xi},
\end{equation}
where $\zeta_{u} = Pr\{ X<u \}$ and $m$ is large enough to ensure that $r_{m} < u$. The return level $r_{m}$ can then be derived by rearranging the terms in Eqn.~(\ref{eqn:retlev}). 

The behavior of the MLEs of the Pareto parameters, including the return level, depends on the value of the shape parameter $\xi$ such that for $\xi<-1$, the MLEs are unlikely to be obtained; for $-1<\xi<-0.5$, the MLEs can be obtained but are not associated with any standard asymptotic property; and if $\xi>-0.5$, the MLEs have the usual asymptomatic distribution. The case $\xi<-0.5$ happens when there are a few extreme events incorporated into the tail analysis, resulting in a distribution with a very short bounded tail. In order to have a shape parameter greater than $-0.5$, more extreme events should be captured by increasing the sample size.
Therefore, if the sample size $n$ is sufficiently large to estimate the tail distribution by using GPD accurately, not only the MLEs of the return levels have an asymptotic property, but also they are normally distributed \cite{evt_09}, \cite{evt_01}, \cite{evt_04}, \cite{samplesize_01}-\cite{evt_14}. Otherwise, MLE is unstable and may provide unrealistic estimates for the model parameters.
\end{proof}
%\vspace{-0.4cm}
\subsection{Dependency Removal Methods}
\label{sec:dependencyremovalmethods}
The input of EVT is necessarily a sequence of i.i.d. random variables. However, the threshold exceedances in the observation sequence are dependent samples inherently, as they occur in groups, i.e., one extremely low power value is likely to be followed by another. To obtain the i.i.d. samples from the sequence of dependent channel data, we either utilize declustering method introduced in \cite{evt_04}, or extract the i.i.d. residuals from ARIMA-GARCH model \cite{garch_01}\nocite{garch_02}\nocite{garch_03}-\cite{garch_04}. 
In the declustering method, we model the tail distribution of the original observation sequence by splitting the observations into multiple clusters separated in time with at least a certain sample gap, each cluster consisting of a group of successive dependent observations. Then, upon extracting the minimum of each cluster, we apply EVT to the i.i.d. cluster minima. This method is based on the fact that the extreme events are close to independent at times that are far enough apart \cite{evt_04}.
On the other hand, in the ARIMA-GARCH model, we apply EVT to the filtered residuals. The filtered residuals represent the randomness of the original sequence that cannot be determined by channel prediction using ARIMA-GARCH. However, in the case that the channel prediction is not the concern, we apply EVT to the original channel data by extracting the i.i.d. samples through the declustering approach.

\subsubsection{Declustering model}
Declustering is a widely adopted method based on the filtering of the dependent observations to obtain a set of threshold excesses that are approximately independent \cite{evt_04}.
In declustering, we specify a threshold $u$ and search for the first observation that falls below $u$ to start the first cluster. The cluster is formed by assigning consecutive values below threshold $u$ to the same cluster. The cluster is deemed to continue until we observe a sample that is above $u$. Upon obtaining an observation above $u$, we let the cluster to conditionally continue storing $r$ more successive values. If any observation within these $r$ samples is found to be below $u$, we continue with the same cluster. Otherwise, we terminate the cluster, and the next sample below $u$ then initiates the next cluster. 
In the declustering method, we let the cluster to store $r$ more values not exceeding the threshold $u$ to overcome the effect of the random noise and deficiency of this method in dependency to the initial choice of threshold \cite{evt_04}.
The optimum values of $u$ and $r$ need to be specified such that the minimum values of the clusters satisfy i.i.d. assumptions and meet the EVT requirements stated in Corollary~\ref{cor:stability}, so that we can model the tail distribution of the received power by using the generalized Pareto distribution.

To obtain the optimum $u$ and $r$ values, we fit the generalized Pareto distribution to the cluster minima for different $u$ and $r$ values , and estimate the corresponding shape and modified scale parameters. Then, based on Corollary~\ref{cor:stability}, we look for the largest $u$ value below which the estimated parameters have a linear relation to the changes in $u$. The linearity relation is evaluated by using the statistical measure, named R squared, denoted by $R^{2}$. This measure is defined as the proportion of variations of one variable explained by the other variable in a linear regression model, taking values between zero and one \cite{linearreg}. Moreover, there is a trade-off in the choice of $r$ value: Too low $r$ creates dependency among the cluster minima due to the consecutive clusters being too close to each other, whereas too high $r$ results in the elimination of some valuable samples that could be extreme in a separated cluster. Therefore, we look for the minimum $r$ value for which the extracted minima are i.i.d., and the estimated Pareto parameters associated to the higher $r$ values do not differ significantly. 

\subsubsection{ARIMA-GARCH model}
\label{sec:arima-garch}
The hybrid ARIMA-GARCH model is used to predict the wireless channel by modeling the mean and variance dependency among the samples. First, ARIMA model is used to model the dependency of the variation in the conditional mean of the time series data, assuming that the variance is constant. Then GARCH model captures the variation in the conditional variance of the data sequence over time. 
The residuals of ARIMA-GARCH filtering then form an i.i.d. sequence, which represents the remaining randomness from the channel prediction, and is used in the analysis of the channel tail statistics by using EVT.

ARIMA$(p,d,q)$ is commonly used to model the conditional mean of a sequence, where $p$ and $q$ stand for the number of autoregressive (AR) and moving average (MA) terms, respectively, and $d$ is the degree of differencing \cite{garch_04}. If the time series data is not stationary, ARIMA parameters estimation starts by transforming the non-stationary process to a stationary one. 
If the non-stationarity is caused by the non-constant variance, power transformation is applied by taking the square root, cube root, or logarithm of the time series. 
%This method is simple but often effective to stabilize the variance across time. 
If non-stationarity is due to the non-constant mean, differencing is applied by creating a new series whose value at time $t$ is the difference between the samples at time $t$ and time $t+d$, i.e., $x(t)-x(t + d)$ \cite{arima_03}. 
%Although there is no limitation on the number of times we apply differencing, it is not recommended to apply it more than twice \cite{arima_01}.
%If the non-stationarity is caused by both non-constant mean and variance then, the power transformation should be applied prior to differencing.
The output of power transformation/differencing is a stationary process denoted by ARIMA($p,0,q$) and given by
\begin{equation}
    r_t = c +  \sum_{i=1}^{p} \theta_{i} r_{t-i} + \sum_{j=1}^{q} \beta_{j} \epsilon_{t-j}
    + \epsilon_{t},
\end{equation}
where $r_{t}$ is the process mean at time $t$ conditioned on the past realizations, $c$ is a constant, $\theta_{i}$ is the $i^{th}$ AR coefficient for $i \in \{1,..., p \}$, $\beta_{j}$ is the $j^{th}$ MA coefficient for $j \in \{1,..., q \}$, and  $\epsilon_{t}$ is the residual at time $t$.
MA and AR coefficients of ARIMA($p,0,q$) can be determined by analyzing the estimated autocorrelation function (ACF) and partial ACF (PACF), respectively \cite{arima_01}-\nocite{arima_02}\cite{garch_06}.

GARCH model characterizes the conditional variance of the predicted ARIMA($p,0,q$) residuals, $\epsilon_{t}$ \cite{garch_01}-\cite{garch_02}, \cite{garch_08}.
%In GARCH, the filtered residuals of ARIMA are decomposed as $\epsilon_{t}=z_{t}\sqrt{\sigma_{t}^{2}}$, where $\{z_{t}\}$ are i.i.d. random variables with zero mean and constant variance \cite{garch_01}-\cite{garch_02}, \cite{garch_08}, and $\sigma_{t}^{2}$ is the variance of the residual at time $t$ conditioned on the previous observations. 
For most practical purposes, the simplest form of the GARCH model, GARCH$(1,1)$, is used, which is given by
\begin{equation}
    \sigma_{t}^{2} = k + \gamma \sigma_{t-1}^{2} + \phi \epsilon_{t-1}^{2} + \psi sgn( -\epsilon_{t-1} ) \epsilon_{t-1}^{2},
\end{equation}
where $k$ is a constant, $\gamma$ and $\phi$ are the coefficients of conditional variance and squared residual at time $t-1$, respectively, and $\psi$ is the coefficient of squared residual at time $t-1$ incorporating the impact of the sign of the residual at time $t-1$. Note that $sgn(-x) = 1$ if $x<0$; $0$ if $x=0$; and $-1$ if $x>0$.
The output of ARIMA-GARCH filtering is $z_{t}$, which is a sequence of i.i.d. residuals with zero mean and the normalized variance one \cite{garch_01}-\nocite{garch_02}\nocite{garch_03}\cite{garch_04}.

The ACF of the residuals, as well as the ACF of the squared residuals, determine the suitability of the fitted ARIMA and GARCH models, respectively. 
%The ACF of the samples measures the correlation among the observations as a function of the time lag to diagnose the validity of ARIMA-GARCH model in removing the linear dependency of the observations by predicting the mean dependency. On the other hand, the ACF of the squared samples quantifies the degree of the autocorrelation between the squared value of the observations to check the validity of the model in detecting the dependency of the samples in the conditional variance \cite{evt_01}, \cite{garch_04}.  
If in both aforementioned ACFs, the significant correlations at lag $0$ are followed by the correlations approximately equal to zero, one may conclude that the parameters of the hybrid ARIMA-GARCH model are determined properly, and the extracted residuals are independent and identically distributed. Otherwise, the estimated ARIMA-GARCH parameters should be revised to remove the additional correlation in the ACFs of the residuals and squared residuals \cite{evt_01}, \cite{garch_04}.

\section{Proposed Channel Modeling Framework}
\label{sec:channel}
%\vspace{0.2cm}
We propose a novel EVT-based channel modeling methodology with the goal of estimating the lower tail statistics of the communication channel in the ultra-reliable regime with very high accuracy and minimum amount of data. 
The methodology consists of the following steps: The sequence of measured received power samples is converted into a sequence of i.i.d samples by removing their dependency via either declustering or ARIMA-GARCH filtering. Upon applying EVT to the resulting sequence of i.i.d. samples, the parameters of the Pareto distribution associated with different thresholds are estimated by using the MLE.
\begin{figure}[ht]
    \centering
    \includegraphics[scale=0.7]{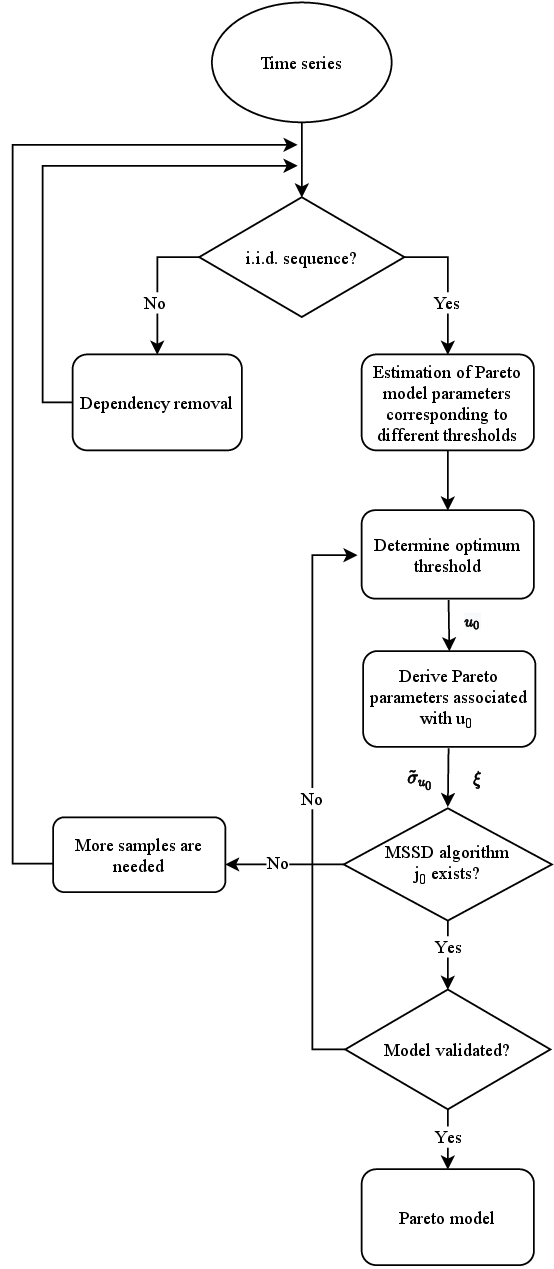}
    \caption{Flowchart of the proposed channel modeling framework.}
    \label{fig:maindiagram}
\end{figure}
%\parbox[t]{\dimexpr\textwidth-\leftmargin}{
%      \vspace{-0.1mm}
%      \begin{wrapfigure}[21]{r}{0.4\textwidth}
%       \centering
%        \vspace{-\baselineskip}
%        \includegraphics[width=\linewidth]{Pareto.pdf}
%        \caption{Flowchart of the proposed channel modeling framework.}
%        \label{fig:maindiagram}
%\end{wrapfigure}
The optimum threshold is then determined by applying mean residual life and parameter stability methods. Next, the stopping conditions are specified to determine the minimum number of samples required for the tail estimation by using GPD. If the number of the collected samples is below this minimum number, the algorithm continues by collecting more data and restarting the process. Otherwise, the validity of the model corresponding to the optimum threshold is assessed by using the probability plots. The proposed algorithm is depicted in Fig.~\ref{fig:maindiagram} and explained in detail next.

\subsection{Dependency removal}
The EVT input needs to be a sequence of i.i.d. random variables, thus, at the first step, the dependency behavior of the input process should be investigated. 
The proposed algorithm starts by calculating the autocorrelation of the input time series sequence. If the autocorrelation at all lags except lag $0$ is effectively $0$, i.e., below a predetermined small threshold value, we skip this step. Otherwise, the dependency removal methods given in Section~\ref{sec:dependencyremovalmethods} are employed: Declustering method is used in estimating the channel tail distribution of the original observed data, while ARIMA-GARCH filtering is used when the channel tail estimation of the residuals extracted from the channel prediction is of interest. 

\subsection{Estimation of Pareto model parameters}
The parameters of the Pareto distribution, $\xi$ and $\tilde{\sigma}_u$, are calculated by using MLE for different threshold values $u$. The threshold value inherently determines the percentage of the data considered as extreme events, and thus, included in the tail statistics. The lower and upper boundaries of the threshold values are chosen as the minimum and average of the i.i.d. input sequence, respectively. 
%\vspace{-0.2cm}
\subsection{Determination of optimum threshold}
\label{sec:threshold}
The suitability of the Pareto model fitted to the tail distribution is significantly affected by the threshold value: Too high a threshold leads to bias due to the immense portion of the data in the tail, while too low a threshold leads to high variance due to a few excesses with which the model is estimated \cite{evt_07}. To determine the optimum threshold, mean residual life and parameter stability methods are employed prior to and after the estimation of the Pareto model parameters, respectively. These two methods are used to validate and complement each other for consistency and robustness in case one of them is not able to determine the optimum threshold.

\subsubsection{Mean residual life method}
\label{sec:meanresplot}
The mean residual life plot consists of the pairs of thresholds and the expected value of the samples exceeding the corresponding threshold, as given by
\begin{equation}
\label{eqn:meanresplot}
    \big\{ \big( u, E(u -X |X < u) \big) : u > X_{min} \big\},
\end{equation} 
where $X = [ X_1, X_2, ..., X_n ]$ is the i.i.d. input sequence and $X_{min}$ is the smallest value in the sequence of $X$ \cite{evt_02}-\cite{evt_04}, \cite{evt_08}. Referring to Corollary \ref{cor:meanexcces}, the best threshold, $u_{0}$, is the highest threshold below which, the mean excess function, i.e., $E(u -X |X < u)$, is a linear function of $u$, based on the $R^2$ measure: Larger value of $R^2$ indicates a better linear fit to the empirical data. Therefore, the mean residual life plot should be a straight line for thresholds $u<u_{0}$, with slope and intercept equal to $\frac{-}\xi{1-\xi}$ and $\frac{\tilde\sigma_{u_{0}} + \xi u_{0}}{1-\xi}$, respectively.
However, if the mean residual life plot behaves almost linearly for all possible thresholds, the decision of optimum threshold is determined by using a parallel approach entitled parameter stability method.

\subsubsection{Parameter stability method}
Parameter stability is a complementary method to the mean residual life plot in the determination of the optimum threshold. 
In this method, first, the scale and shape parameters of the GPD fitted to the tail distributions corresponding to different thresholds are extracted. Then, based on Corollary \ref{cor:stability}, the optimum threshold, $u_{0}$, is the highest threshold below which the estimated scale and shape parameters are linear function of $u$, according to the $R^2$ measure.
%\vspace{-0.3cm}
\subsection{Minimum Sample Size Determination (MSSD) Algorithm}
\label{sec:sample size}
MSSD is developed with the goal of determining the stopping conditions for the collection of enough samples used in modeling the tail statistics.
Although the increase in sample size provides a more accurate estimation of the tail statistics, collecting a large number of samples might be expensive, time-consuming, and even impossible in many experiments \cite{evt_09}.
The MSSD algorithm is a smart enumeration based algorithm, mainly inspired by the methodology given in \cite{evt_09}, which has been proposed to determine the minimum sample size required to model the distribution of maxima/minima in a time sequence by using GEV distribution. We adapt this method for the determination of the minimum number of samples required to estimate the distribution of the values exceeding a given threshold in a time series by using GPD.

MSSD Algorithm is based on the normality of the return levels given in Corollary~\ref{cor:returnlevel}: If the sample size $n$ is sufficiently large to estimate the tail distribution of values exceeding threshold $u$ by using GPD($\xi$, $\tilde \sigma_{u}$), then, the return level, $r_{m}$, is normally distributed. However, in most of the experiments, it is not practical to check the normality assumption of the return levels as there exists only one or a few time series. To tackle this problem, standard bootstrap method with replacement is applied to the observations with the goal of generating more sets for assessing the distribution of the return levels \cite{evt_09}. MSSD Algorithm consists of two main parts:
\begin{enumerate}
    \item \textit{Two-step bootstrapping process:} Bootstrapping is carried out in two steps to achieve a better bootstrapping performance in terms of the estimation accuracy, as stated in \cite{evt_09}, \cite{samplesize_01}. In the first step, the bootstrap is performed on the original observations to generate $M-1$ data sets, each with sample size $n$, same as the size of the original data set. In the second step, sample sets of different sizes are extracted from each of the $M$ data sets in order to obtain the return levels associated with different sample sizes, performing $K-1$ bootstrapping on these data sets. Therefore, at the end of the two-step bootstrapping, $M \times K$ data sets with different lengths are generated.
    \item \textit{Normality assessment of the return levels:} Upon extracting the shape and scale parameters of the GPD fitted to the tail distribution of each sample set, we calculate the corresponding return levels for assessing their normality by using Anderson-Darling (AD) normality test. AD test is used instead of Shapiro-Wilk test in \cite{evt_09} since AD test is the most powerful empirical distribution function (EDF) test with higher focus on the tail of the distribution, compared to the other existing normality tests, and more appropriate for assessing the normality of the data with sample sizes greater than $50$ \cite{normaltest_01}. After applying AD test, we extract the $p$-values, where $p$-value is defined as the index measuring the strength of the evidence against the null hypothesis $H_0$ that the data follows the normal distribution. Then, we compare these $p$-values with the critical value $\alpha$, which is defined as the probability of rejecting the null hypothesis when the null hypothesis is true, i.e., probability of making a wrong decision, to (i) check whether there are enough samples for modeling the tail distribution, (ii) determine the minimum number of samples required to estimate the tail distribution, and (iii) approximate the gain attained by increasing the sample size more than its minimum required value. 
\end{enumerate}

\begin{algorithm}[b!]
\caption{\textbf{Minimum Sample Size Determination (MSSD) Algorithm}}
\label{alg:returnlevel}
\begin{algorithmic}
\STATE \textbf{Input}: {\boldm $X$} $= [X_{1}, X_{2},...,X_{n}]$, $\alpha$, $M$, $K$, and $n_{0}$;
\STATE \textbf{Output}: $j_{0}$, $G_{j}$, $j \in \{j_{0},...,n\}$;
\end{algorithmic}
\begin{algorithmic}[1]
\STATE {\boldm $x$}$_1$ $=(x_{11},...,x_{1n})$ $=$ {\boldm $X$};
\FOR{m=1:M}
\IF{$m \neq 1$}
 \STATE obtain bootstrap samples of {\boldm $x$}$_1$, as
 {\boldm $x$}$_m$ $= (x_{m1}, ..., x_{mn})$;
\ENDIF
\FOR{$j=n_{0}:n$}
\STATE {\boldm $y$}$_{mj}^{1}$ $=(x_{m1}, ..., x_{mj})$; 
\FOR{$k=1:K$}
\IF{$k \neq 1$}
\STATE obtain bootstrap samples of {\boldm $y$}$_{mj}^{1}$, as {\boldm $y$}$_{mj}^{k}$;
\ENDIF
\STATE \text{calculate $r_{mj}^{k}$ using parameters of GPD fit for {\boldm $y$}$_{mj}^{k}$};
\ENDFOR
\STATE {\boldm $d$}$_{mj}$ $= (r_{mj}^{1},...,r_{mj}^{K})$;
\STATE store $p$-value of AD normality test on {\boldm $d$}$_{mj}$ in $p_{mj}$;
\ENDFOR
\ENDFOR
\FOR{$j=n_{0}:n$}
\STATE calculate the mean of $p$-value, $\bar{p}_{j}$;
\STATE calculate the standard deviation of $p$-value, $s_{j}$;
\ENDFOR
\IF{$\exists j_{0} \in [n_{0},n]$  \textbf{such that} $\bar{p}_{j} - t^{*} s_{j}>\alpha$ for all $j \in [j_{0},n]$, \textbf{and} $\nexists i <j_{0}$  \textbf{such that} $\bar{p}_{j} - t^{*} s_{j}>\alpha$ for all $j \in [i,n]$}
\FOR{$j=(j_{0}+1):n$}
\STATE $G_{j} = s_{j}-s_{j_{0}}$;
\ENDFOR
\RETURN $j_{0}$, $\{G_{j_{0}+1},...,G_{n}\}$;
\ELSE
\RETURN no feasible channel tail estimation exists;
\ENDIF
\end{algorithmic}
\end{algorithm}
MSSD Algorithm, given in Algorithm \ref{alg:returnlevel}, is described in detail as follows. The inputs of the algorithm are the observation sample set {\boldm $X$}; the critical value $\alpha$, which is used as a threshold for assessing the $p$-values of the AD normality test; the number of new generated data sets in the first and second steps of the bootstrapping process denoted by $M$ and $K$, respectively; and the minimum sample size $n_{0}$ for which the estimated shape parameter $\xi$ of the GPD fitted to the values exceeding the optimum threshold $u_{0}$ is greater than $-0.5$, based on Corollary~\ref{cor:returnlevel}. The algorithm starts by storing the observation samples in a new vector {\boldm $x$}$_1$ for more convenient notation (Line~$1$). Then, the first step of bootstrapping with replacement is performed on the original data set to generate $M-1$ sample sets, each of size $n$, denoted by {\boldm $x$}$_m$, $m\in\{2,...,M\}$ (Lines~$2-4$). Next, to perform the second step of bootstrapping, upon extracting the first $j$ samples from each data set {\boldm $x$}$_m$, and assigning them to a new vector denoted by {\boldm $y$}$_{mj}^{1}$, for $j\in\{n_{0},...,n\}$ and $m\in\{1,...,M\}$ (Lines~$6-7$), $K-1$ bootstrapping with replacement is performed on {\boldm $y$}$_{mj}^{1}$, storing each newly generated data set in the corresponding {\boldm $y$}$_{mj}^{k}$, where $k\in\{2,...,K\}$ (Line~$10$). Then, the GPD fit is applied on the values below the optimum predefined threshold of the data set {\boldm $y$}$_{mj}^{k}$, and the corresponding scale and shape parameters are used in the calculation of the corresponding return levels $r_{mj}^{k}$ from Eqn.~(\ref{eqn:returnlevel}) (Line~$12$). The return levels associated to the same sample size are placed into the vector {\boldm $d$}$_{mj}$, for $j\in\{n_{0},...,n\}$ and $m\in\{1,...,M\}$ (Line~$14$). Next, the normality assumption of the return levels is checked by applying the AD normality test on {\boldm $d$}$_{mj}$, storing the $p$-values of the test in $p_{mj}$ (Line~$15$). To enhance the performance of the bootstrapping process and decrease the test error in the estimation of the $p$-value, the average of $M$ $p$-values corresponding to the same sample size $j$, denoted by $\bar{p}_{j}$, is calculated by using $ \bar{p}_{j} = ({1}/{M}) \sum_{m=1}^{M}p_{mj}$, for $j\in\{n_{0},...,n\}$ (Line~$19$). However, the expected $p$-value is not exactly $\bar{p}_{j}$, and typically takes a value within the range of $[\bar{p}_{j} - t^{*} s_{j},\bar{p}_{j} + t^{*} s_{j}]$, where $t^{*}$ is the critical value obtained from a $t$-distribution with $M-1$ degrees of freedom, and $s_{j}$ is the standard deviation of the $p$-value corresponding to the same sample size $j$, calculated by $s_{j} =\sqrt{({1}/{M}) \sum_{m=1}^{M}(p_{mj}-\bar{p}_{j})^{2}}$ (Line~$20$). 
Based on Corollary~\ref{cor:returnlevel}, if there exists a sample size $j_{0}$ such that the return level is normally distributed for all sample sizes greater than $j_{0}$, i.e., $p$-value, or more explicitly its lower bound denoted by $\bar{p}_{j} - t^{*} s_{j}$, is greater than $\alpha$ for $j>j_{0}$, and there exists no sample size $i$ less than $j_{0}$, such that $\bar{p}_{j} - t^{*} s_{j} >\alpha$ for $j>i$, then, $j_{0}$ is the minimum plausible sample size with which we can estimate the tail distribution by using GPD (Line~$22$). Assuming a small quantity, such as $0.01$ or $0.05$, for the critical value $\alpha$, we are $100(1-\alpha)\%$ confident that $j_{0}$ is the minimum number of samples required to obtain a properly normally distributed return level, and therefore, estimate the tail distribution by using GPD. Then, the algorithm determines the gain attained by increasing the sample size beyond the sufficient number $j_{0}$ (Lines~$23-24$). 
The gain is defined as the difference between the standard deviations $s_j$ and $s_{j_0}$ corresponding to the sample sizes $j$ and $j_0$, respectively, where $j\in\{j_{0},...,n \}$. In other words, the gain determines how much we can expand the range of values for the expected $p$-value, i.e., $[\bar{p}_{j} - t^{*} s_{j},\bar{p}_{j} + t^{*} s_{j}]$, while still keeping the $p$-value above the critical value $\alpha$, i.e., the distribution of the return levels normal. The algorithm terminates by returning $j_0$, if it exists, and the obtained gains for the sample sizes greater than $j_{0}$ (Line~$26$). Otherwise, the algorithm terminates stating that no feasible channel tail estimation exists since the existing samples are not sufficient to characterize the tail statistic (Lines~$27-28$). Then, the sample size needs to increase by at least $0.1 \times n_{0}$ samples. The complexity of the MSSD algorithm is $O(nMK)$, since it bootstraps the sequence with sample size $n$ for $M$ iterations with complexity $O(nM)$, and then, by extracting a portion of samples from the newly generated sequences, MSSD algorithm bootstraps them for $K$ iterations. 
%\vspace{-0.3cm}
\subsection{Model Validation}
\label{sec:modelchecking}
The model validity is assessed by using the probability plots in extreme value analysis. The probability plots are graphical techniques used to investigate whether the output of the proposed model fits well to the actual values. Two kinds of probability plots are used for validation: Probability/Probability (PP) plot and Quantile/Quantile (QQ) plot. 

In the PP plot, the CDF of the received power is plotted versus the modeled CDF by GPD. \textbf{PP plot} includes the pairs 
%\vspace{-0.5cm}
\begin{align}
\label{eqn:probilityplot}
\begin{split}
   \{ ({i}/{k+1}), H_{u}(y_{i}); i = 1,...,k \},\\
   \hfill \hspace{0.3cm}
   H_{u}(y_{i}) = 1 - \big( 1+\frac{{\xi}y_{i}}{\tilde{\sigma}_{u}} \big) ^{-1/{\xi}},
\end{split}
\end{align}
where $y_{i}$ is the difference between the threshold $u$ and $X_{i}$ value exceeding the threshold such that $y_{1} \leq ... \leq y_{k}$; $k$ is the total number of excess values; $H_{u}$ is the estimated Pareto model fitted to the CDF of the threshold excesses for threshold $u$ with the associated shape and scale parameters $\xi$ and $\tilde{\sigma}_u$, respectively \cite{evt_04}. 
For any one of the $y_{i}$, exactly $i$ out of the $k$ observations have a value less than or equal to $y_{i}$, so the empirical probability of an observation
being less than or equal to $y_{i}$ is ${i}/{k+1}$. Note that $1$ in the denominator is a slight adjustment that is usually made to avoid reaching the CDF exactly equal to $1$ \cite{evt_04}.

In the QQ plot, empirical and modeled quantiles are plotted against each other. \textbf{QQ plot} consists of the pairs
\begin{align}
    \label{eqn:quantileplot}
    \begin{split}
        \{X_{i},({H}_{u}^{-1}(i/(k+1)), i=1,...,k \},\\
        \hfill \hspace{0.2cm}
        {H}_{u}^{-1}(z_{i}) = u + \frac{\tilde{\sigma}_{u}}{\xi}\big[1 - z_{i}^{-{\xi}} \big],
    \end{split}
\end{align}
where $u$ denotes the threshold; $\xi$ and $\tilde{\sigma}_u$ are the associated shape and scale parameters of the GPD fitted to the values below threshold $u$, respectively; $z_{i}$ is the probability whose associated quantile is of interest; $X_{i}$ is the $i^{th}$ input sample exceeding the threshold such that $X_{1}\leq...\leq X_{k}$; and $k$ is the number of values exceeding $u$. For any one of the $X_{i}$, exactly $i$ out of the $k$ observations have a value less than or equal to $X_{i}$. 

If the generalized Pareto distribution appropriately estimates the extreme values exceeding a threshold $u$, then, nearly all data points in both probability plots lie on the unit diagonal, i.e., the $45^\circ$ line \cite{evt_01}, \cite{evt_04}.

\section{Performance Evaluation}
\label{sec:performance evaluation}
The goal of this section is to evaluate the performance of the proposed channel modeling algorithm in determining the optimum threshold and optimum stopping conditions on the number of samples required to estimate the channel tail, and compare its performance with the traditional extrapolation-based methods in the estimation of the channel tail statistics. The traditional extrapolation-based approach first estimates the distribution of the existing channel data for the reliability order of $10^{-3}-10^{0}$ PER \cite{vehicular_01}-\cite{vehicular_02}, and then, extrapolates the fitted distribution toward the ultra reliable region of $10^{-9}-10^{-6}$ PER \cite{urllc_02}.

We have collected channel measurement data within the engine compartment of Fiat Linea under various engine and driving scenarios at $60$ GHz by using a Vector Network Analyzer (VNA) (R$\And$S$\textsuperscript{\textregistered}$ ZVA$67$). The transmitter and receiver are attached to the VNA ports through the R$\And$S$\textsuperscript{\textregistered}$ ZV-Z$196$ port cables with $610$ mm length and maximum $4.8$ dB transmission loss. Both transmitter and receiver antennas are horn antennas operating between $50$-$75$ GHz with a nominal $24$ dBi gain and $12^\circ$ vertical beamwidth. The antennas are connected to the waveguide operating at the frequency span of $50$-$65$ GHz, with insertion loss $0.5$ dB and impedance $50$ $\Omega$. 
The locations of the transmitter and receiver antennas are selected out of the possible locations for the wireless sensors located within the engine compartment, namely locations $5$ and $13$ in \cite[Fig.~$1$]{vehicular_02}, and shown in Fig.~\ref{fig:fiatlinea}.
\begin{figure}[ht]
    \centering
    \includegraphics[width=.7\linewidth]{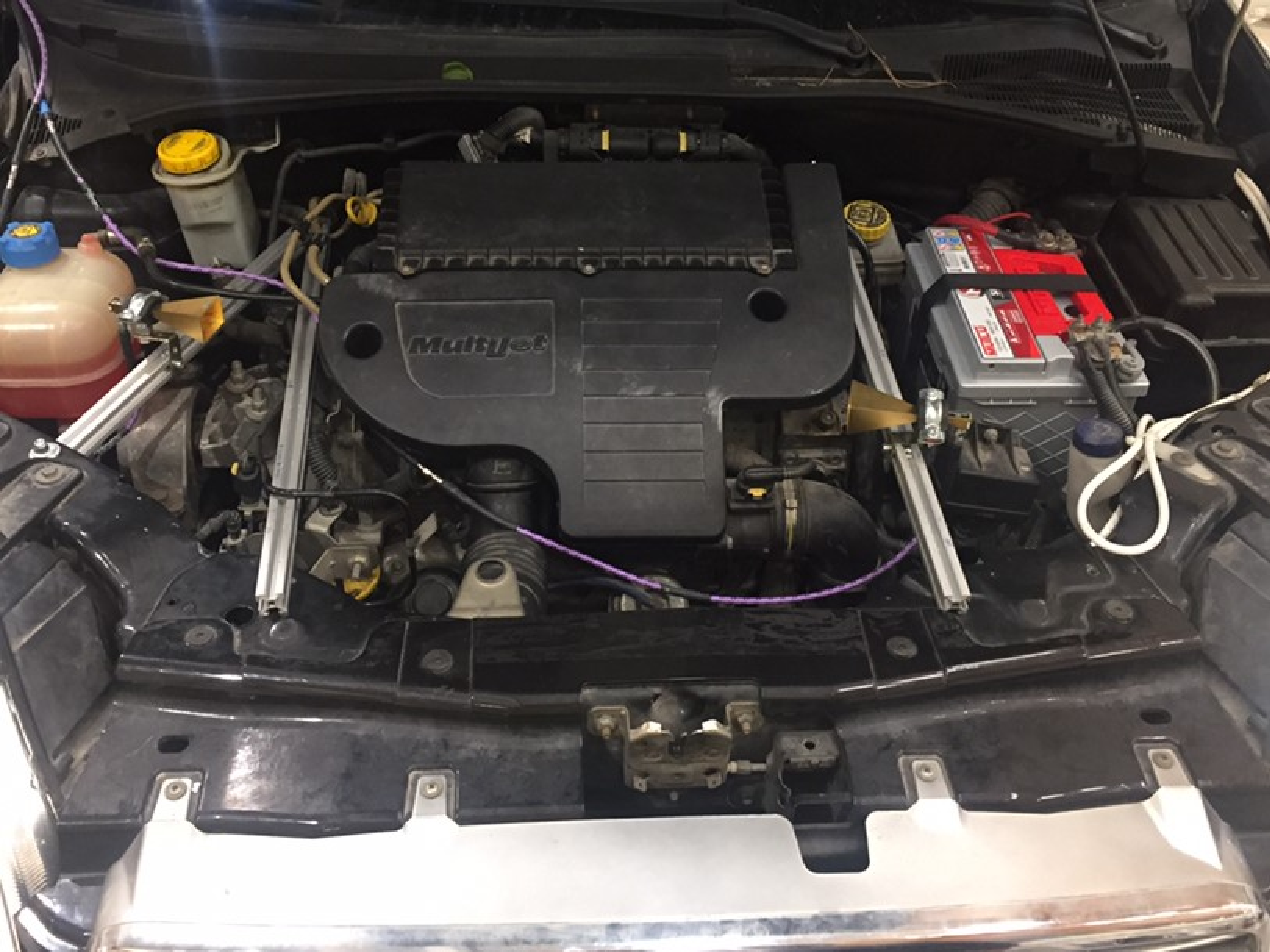}
    \caption{Transmitter and receiver antennas in the engine compartment.}
    \label{fig:fiatlinea}
\end{figure}
%\parbox[t]{\dimexpr\textwidth-\leftmargin}{%
%      \vspace{-0.5mm}
%      \begin{wrapfigure}[9]{r}{0.4\textwidth}
%        \centering
%        \vspace{-\baselineskip}
%        \includegraphics[width=\linewidth]{fiatlinea.jpg}
%        \caption{Transmitter and receiver antennas in the engine compartment.}
%        \label{fig:fiatlinea}
%\end{wrapfigure}

The data were collected while the car is moving on the asphalt and stone roads at Koc University campus. The driving scenarios for the car include pushing the gas paddle while the car is static, moving up a ramp, and driving on the flat road. About $10^{6}$ successive samples are captured for $30$ minutes with time resolution of $2$ ms. We use \textsc{MATLAB} for the implementation of the proposed algorithm as well as the traditional extrapolation-based approach \cite{urllc_02}.

In the following, first, we provide the numerical results in the determination of the optimum threshold over which the tail statistics are derived based on two approaches, ARIMA-GARCH and declustering, and then validate the tail model corresponding to the optimum threshold using probability plots in Section~\ref{sec:PerformanceOptimumThreshold}. Next, in Section~\ref{sec:PerformanceSampleNumber}, we present the performance evaluation of the proposed MSSD algorithm in determining the minimum number of samples required to estimate the channel tail statistics. Finally, we compare the performance of the proposed methodology to that of the traditional extrapolation-based technique in the estimation of the channel tail statistics in Section~\ref{sec:PerformanceComparison}.
%\vspace{-0.3cm}
\subsection{Optimum Threshold Determination}
\label{sec:PerformanceOptimumThreshold}
\subsubsection{ARIMA-GARCH approach}
ARIMA-GARCH filtering is used to predict the channel by its mean and variance conditioned on the past observations, and remove the dependency among the samples. Upon applying EVT on the i.i.d. residuals of ARIMA-GARCH, the best threshold is determined by using the mean residual life plot and the parameters stability methods. The probability plots are then shown to validate the GPD model fitted to the channel tail distribution. Please note that due to the stationarity of the process, the $d$ parameter of ARIMA-GARCH is $0$, and the best performance of ARIMA-GARCH at which the correlations of the residuals are almost $0$ for any time lag greater than $0$ is achieved when $p$ and $q$ are $2$ \cite{arima_01}-\cite{arima_02}.

Fig.~\ref{fig:meanresplot} shows the mean excess of the filtered residuals at different thresholds for ARIMA-GARCH filtering based on the mean residual life method. Mean excesses quantify the expected values of the residuals exceeding a given threshold, where the threshold varies from $-3.85$ dBm to $-0.47$ dBm.
The lower and upper boundaries of the threshold are chosen such that, $0.3\%$ and $30\%$ of the extremely low-value residuals are in the tail, respectively.
Since the mean excess increases linearly as threshold increases for all threshold values, i.e., $R^2>0.95$ for all thresholds, it is not possible to accurately distinguish the optimum threshold value based on Corollary~\ref{cor:meanexcces}, requiring the use of the parameter stability method. 

\begin{figure}[!htb]
   \begin{minipage}{0.48\textwidth}
     \centering
     \includegraphics[width=.7\linewidth]{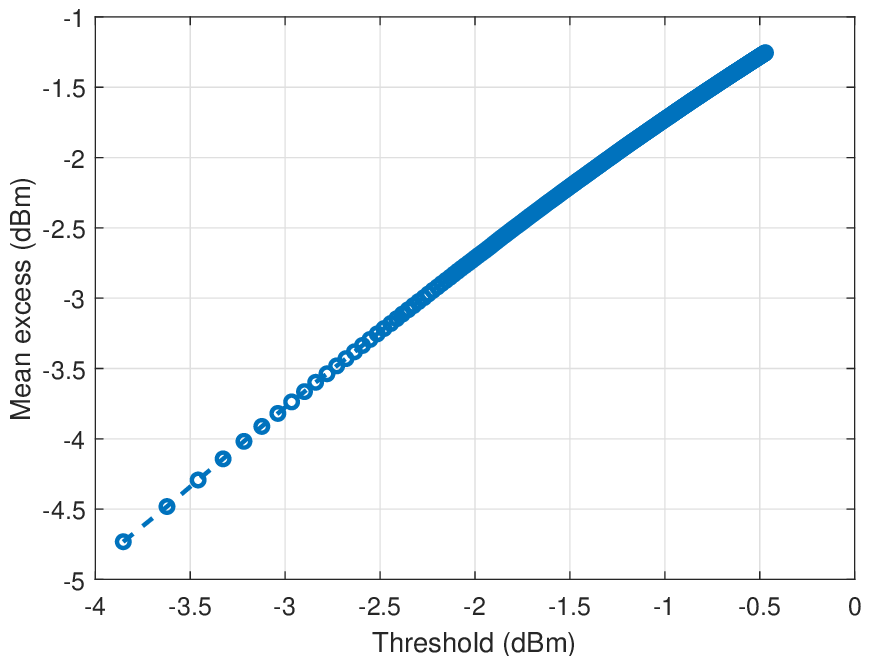}
     %\vspace{-0.2cm}
     \caption{Mean excesses at different thresholds for ARIMA-GARCH filtering.}\label{fig:meanresplot}
   \end{minipage}\hfill
   \vspace{0.3cm}
   \begin{minipage}{0.48\textwidth}
     \centering
     \includegraphics[width=.7\linewidth]{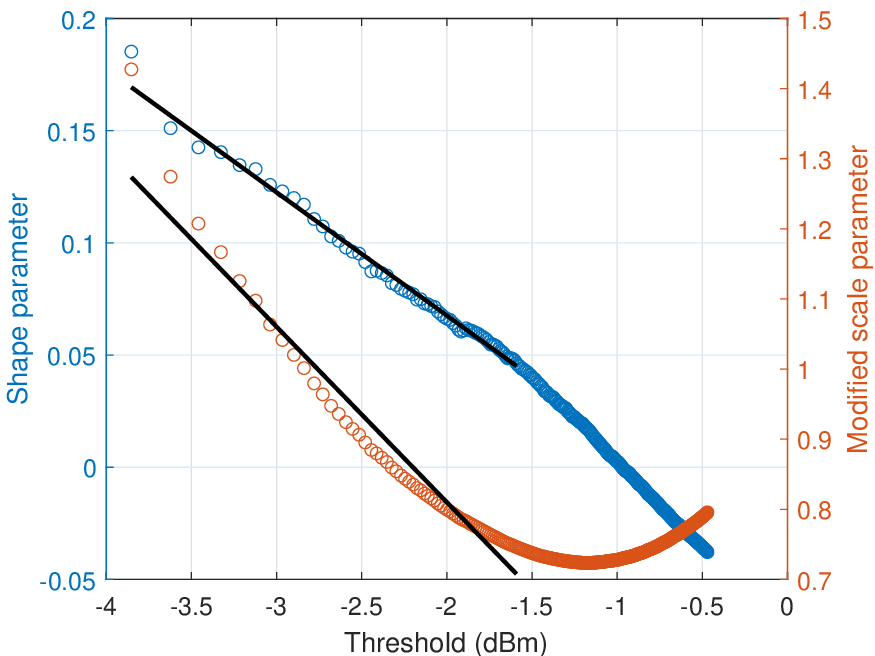}
     %\vspace{-0.2cm}
     \caption{Shape and modified scale parameters of GPD at different thresholds for ARIMA-GARCH filtering. The black lines are the fitted linear regression models.}\label{fig:stabilityplot}
   \end{minipage}
\end{figure}
Fig.~\ref{fig:stabilityplot} shows the shape and modified scale parameters of the fitted GPD model at different thresholds on the i.i.d. input sequence obtained by ARIMA-GARCH filtering for the parameter stability method.
As opposed to Fig.~\ref{fig:meanresplot}, Fig.~\ref{fig:stabilityplot} can be used to determine the optimum value below which both shape and modified scale parameters change linearly with respect to the threshold. The optimum threshold, $u_{0}=-1.6$ dBm, is in fact, the maximum threshold value for which the corresponding regression models fitted to the estimated parameters at $u<u_{0}$ results in $R^2$ value greater than $0.95$. The black lines illustrate the linear regressions fitted to the parameters at the optimum threshold and the thresholds below that. $\sigma^{*}$ and $\xi$ corresponding to the optimum threshold $-1.6$ dBm are $0.746$ and $0.046$, respectively.

Fig.~\ref{fig:probplotsArima} shows the probability plots, including PP plot and QQ plot, for ARIMA-GARCH approach. This figure is used to validate the performance of the fitted Pareto model associated to the optimum threshold derived from Fig.~\ref{fig:stabilityplot}. The black line is the diagonal line for diagnosing whether the values obtained by the Pareto model are in a good agreement with those of the empirical model. 
The black line has slope $1$ and intersections $0$ and $-16$ in the PP plot and QQ plot, respectively. The linearity of the plots is reasonable, but not perfect, especially at the left side of the QQ plot. Fig.~\ref{fig:ppplot} shows that the CDF values obtained by the Pareto model are identical to the corresponding empirical values. However, Fig.~\ref{fig:qqplot} shows that the Pareto model provides better performance at high quantile values than the quantile values lower than $-7$ dBm, mainly due to the deviation of the i.i.d. residual distribution from the normal distribution. In spite of the divergence of the modeled from the empirical in Fig.~\ref{fig:qqplot}, since Fig.~\ref{fig:ppplot} illustrates a perfect agreement between the empirical and modeled values, one can still conclude that the fitted Pareto distribution models the empirical i.i.d. residuals obtained by ARIMA-GARCH filtering accurately \cite{evt_01}.
%\vspace{-0.5cm}
\begin{figure}[t]
\centering
\captionsetup[subfigure]{labelformat=empty}
     \begin{center}
        \subfloat[(a)]{%
            \label{fig:ppplot}
            \includegraphics[width=.7\linewidth]{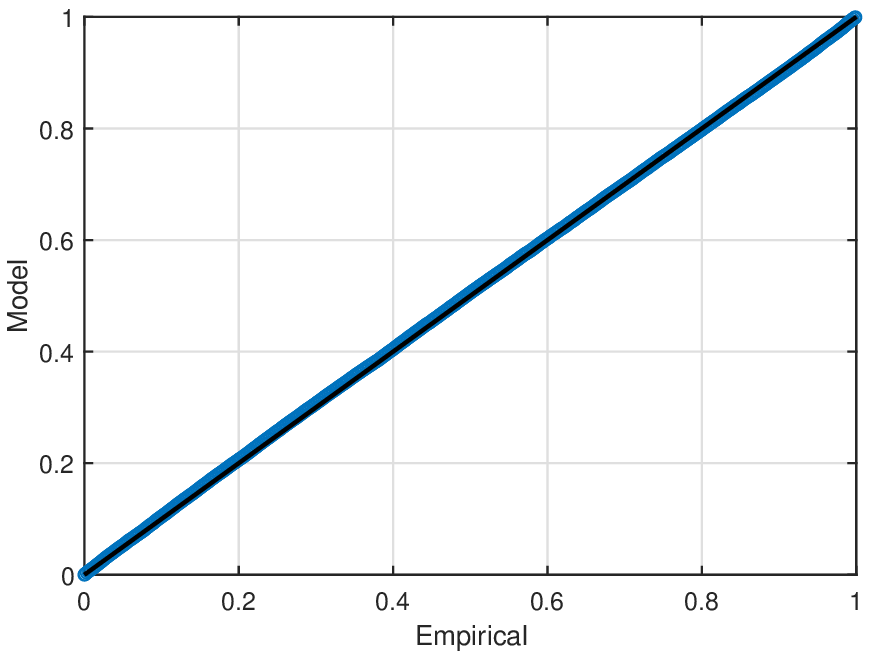}
        }\\%
        \subfloat[(b)]{%
            \label{fig:qqplot}
            \includegraphics[width=.7\linewidth]{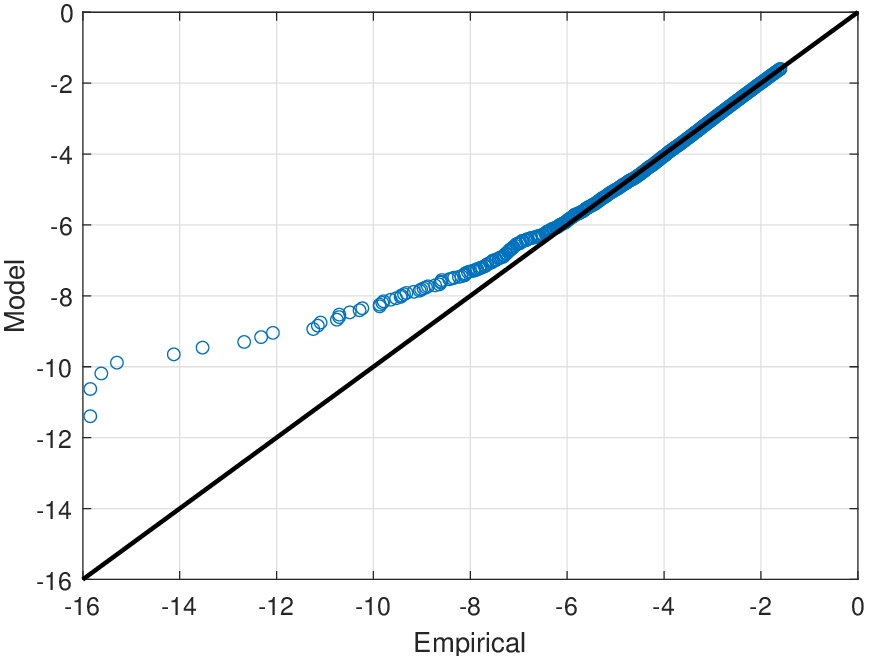}
        }\\ 
    \end{center}
    %\vspace{-5mm}
    \caption{Probability plots for ARIMA-GARCH filtering: (a) PP plot, and (b) QQ plot. The black lines are diagonal lines assessing the goodness of fit.}
   \label{fig:probplotsArima}
\end{figure}
%\vspace{-0.2cm}
\subsubsection{Declustering approach}\label{sec:performance-declustering}
Declustering is used to remove the dependency among the samples and obtain i.i.d. observation for the EVT input. First, the optimum values for the threshold $u$ and the minimum gap $r$ between the samples are determined by using mean residual life plot and parameter stability method. Then, the probability plots are used to validate the accuracy of the GPD model fitted to the channel tail distribution.

Fig.~\ref{fig:meanresplotdeclustering} shows the mean excess of the observations at different thresholds and minimum gaps between the clusters for the declustering approach based on the mean residual life method. The expected values of the samples exceeding a given threshold are plotted against the threshold, where the threshold varies from $-25$ dBm to $-10$ dBm. When there is no minimum gap restriction between the clusters, i.e., $r=0$, the mean excess increases linearly as
threshold increases for all threshold values, similar to Fig.~\ref{fig:meanresplot}. However, by choosing $r>0$, the mean residual life plot is linear in threshold $u$, for $u<-19$ dBm, where $-19$ dBm is the largest threshold value for which $R^{2}>0.95$ for all $r>0$. Therefore, according to the mean residual life plot for the declustering method, the optimum threshold and the minimum $r$ values are $-19$ dBm and $1$, respectively.    

\begin{figure}[ht]
\centering
     \begin{center}
            \includegraphics[width=.7\linewidth]{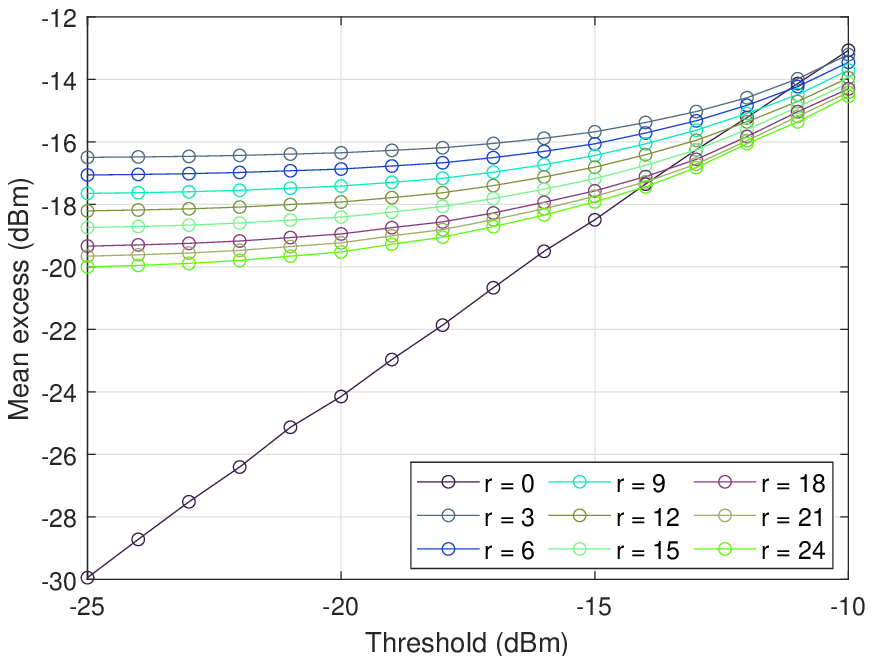}
    \end{center}
    %\vspace{-0.5cm}
    \caption{Mean excesses at different thresholds and minimum gaps between the clusters for declustering approach.}%
   \label{fig:meanresplotdeclustering}
\end{figure}
Fig.~\ref{fig:paramsvsth} shows the shape and modified scale parameters of the fitted GPD model on the i.i.d. sequence obtained by the declustering approach at different $u$ and $r$ values for the parameter stability method.
By choosing $r$ value greater than $15$, the estimated Pareto parameters are linear in $u$ for $u<-19$ dBm. Although in Fig.~\ref{fig:meanresplotdeclustering}, it is possible to attain linearity assumption of the mean excesses for $u<-19$ dBm with $r$ values less than $15$, the linearity is observed only for $r>15$ at $u<-19$ dBm in Fig.~\ref{fig:paramsvsth}. Therefore, $r$ value should be at least $15$. Additionally, $-19$ dBm is the minimum threshold for which the $R^2$ value of the regression model fitted to the estimated parameters is greater than $0.95$.
As a result, considering Figs.~\ref{fig:meanresplotdeclustering} and~\ref{fig:paramsvsth} together, the optimum values of $u$ and $r$ are $16$ and $-19$ dBm, respectively. Also, for these optimum values, the corresponding shape and modified scale parameters are $0.108$ and $2.06$, respectively.

\begin{figure}[ht]
\centering
\captionsetup[subfigure]{labelformat=empty}
     \begin{center}

        \subfloat[(a)]{%
            \label{fig:scalevsth}
            \includegraphics[width=.7\linewidth]{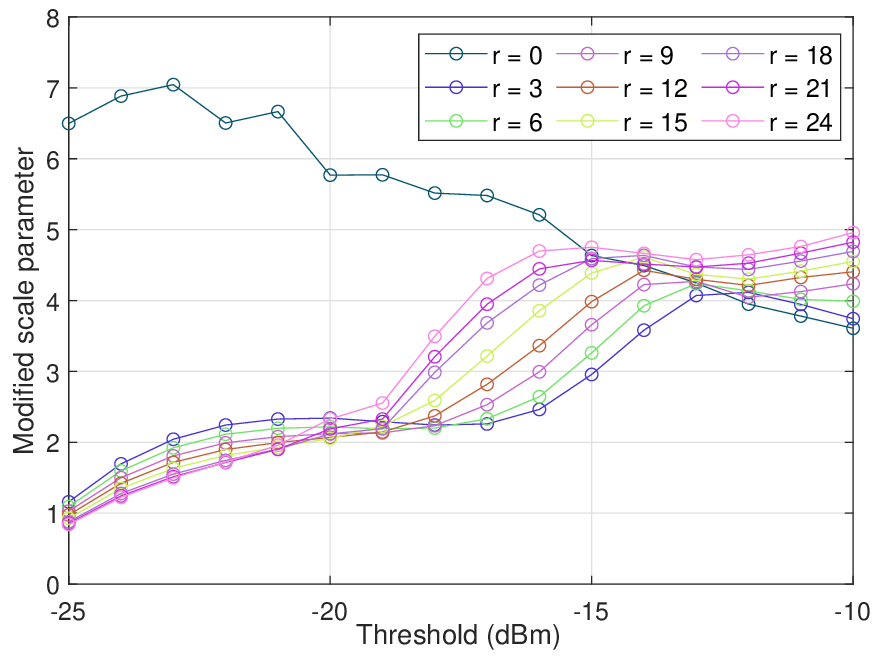}
        }\\%
        \subfloat[(b)]{%
            \label{fig:shapevsth}
            \includegraphics[width=.7\linewidth]{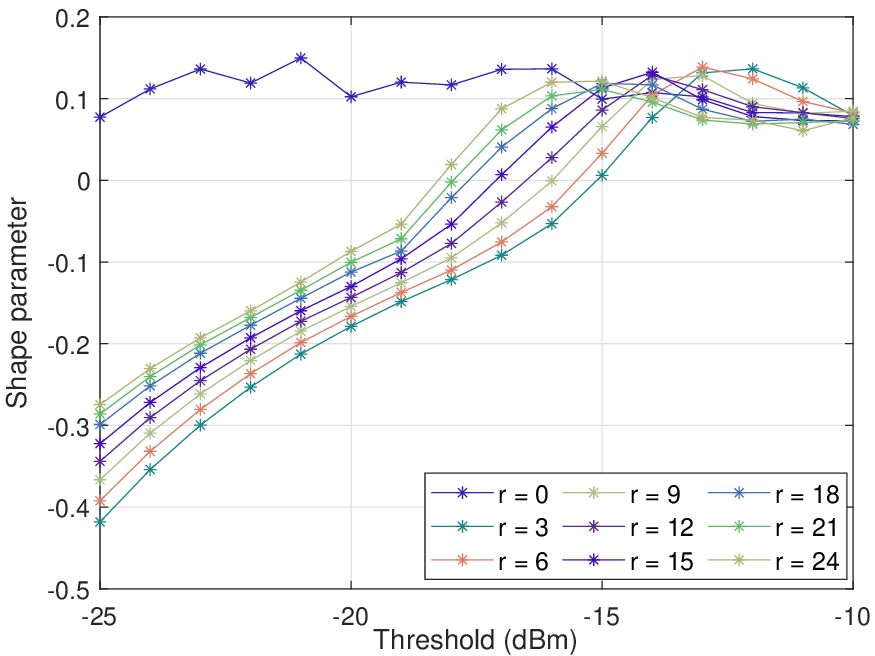}
        }\\
    \end{center}
    %\vspace{-0.3cm}
    \caption{Shape and modified scale parameters of GPD at different thresholds and minimum gaps between the clusters for declustering method: (a) Modified scale parameter, and (b) Shape parameter.}%
  \label{fig:paramsvsth}
\end{figure}

Fig.~\ref{fig:probplotsdeclustering} shows the probability plots for the declustering approach. This figure is used to validate the performance of the Pareto model corresponding to the threshold value evoked from Figs.~\ref{fig:meanresplotdeclustering} and \ref{fig:paramsvsth}. The black line is the diagonal line for diagnosing the goodness of fit of the GPD to the empirical values.
The black lines have both slope $1$, but intersections $0$ and $-55$ for the PP plot and QQ plot, respectively. Both PP and QQ plots reveal the robust performance of the generalized Pareto model applied on the output of the declustering approach.
Similar to Fig.~\ref{fig:ppplot}, Fig.~\ref{fig:ppplotdeclustering} shows that the values modeled by GPD for the declustering approach are almost identical to that of the empirical values. However, as opposed to Fig.~\ref{fig:qqplot}, Fig.~\ref{fig:qqplotdeclustering} shows that the QQ plot of the declustering approach fits to the diagonal line for all the quantile values. Therefore, the Pareto model obtained by the declustering approach models the empirical values with a much better performance than that of the ARIMA-GARCH approach. 
\begin{figure}[ht]
\centering
\captionsetup[subfigure]{labelformat=empty}
     \begin{center}
        \subfloat[(a)]{%
            \label{fig:ppplotdeclustering}
            \includegraphics[width=.7\linewidth]{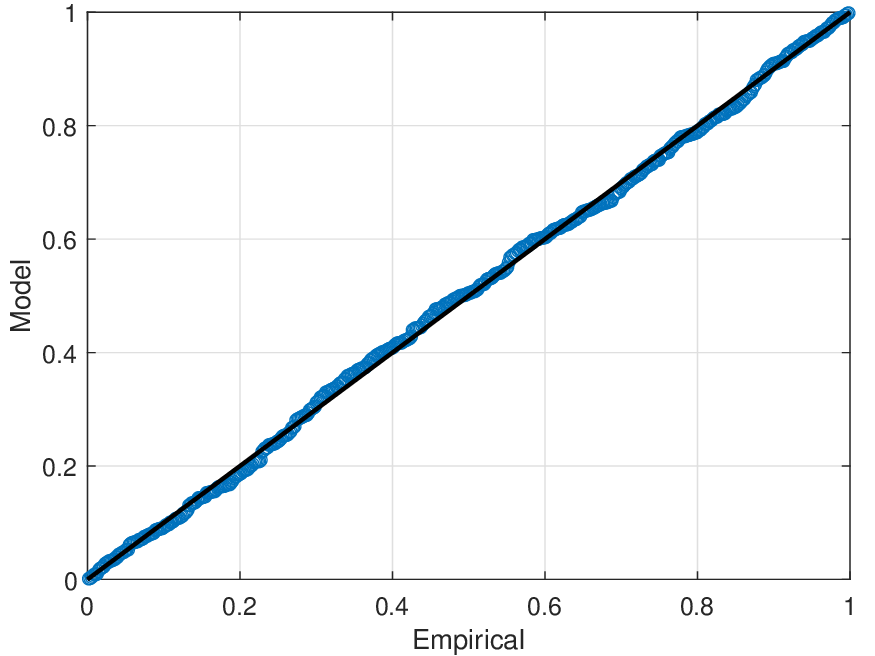}
        }\\%
        \subfloat[(b)]{%
            \label{fig:qqplotdeclustering}
            \includegraphics[width=.7\linewidth]{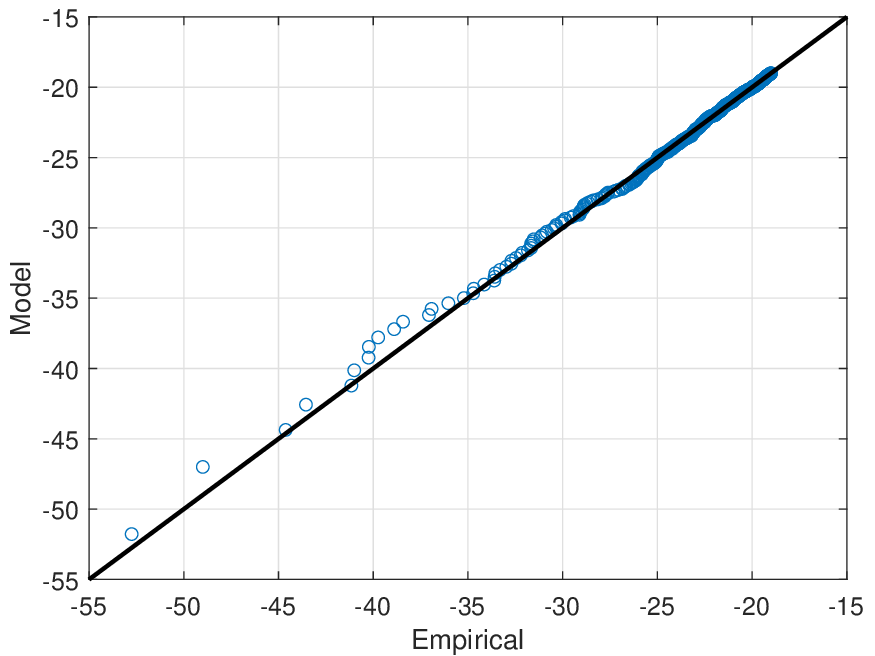}
        }\\
    \end{center}
    %\vspace{-0.5cm}
    \caption{Probability plots for declustering method: (a) PP plot, and (b) QQ plot. The black lines are diagonal lines assessing the goodness of fit.}%
   \label{fig:probplotsdeclustering}
\end{figure}

%\vspace{-0.7cm}
\subsection{Minimum Required Number of Samples for Channel Tail Estimation}
\label{sec:PerformanceSampleNumber}
Fig.~\ref{fig:samplesize} shows the p-value of the AD normality test, as explained in detail in Section~\ref{sec:sample size}, for different number of samples. This figure is used to determine the minimum number of samples required to model the channel tail statistics. 
The horizontal line is at the critical value $\alpha = 0.05$, and the vertical line is at sample size equal to $2 \times 10^5$. For the sample size greater than $2 \times 10^5$, all the $p$-values are above the horizontal line, indicating that the sufficient number of samples required to estimate the tail statistics is $2 \times 10^5$. The gain we attain by increasing the sample size from $2 \times 10^5$ to $9 \times 10^{5}$, is only $0.08$ in terms of the standard deviation of the $p$-values. As a result, although it is required to capture $10^6$ samples or more for modeling the events occurring once in million for URLLC, MSSD allows the estimation of the tail statistics with lower number of samples.
%\begin{figure}[t]
%    \centering
%    \includegraphics[width=0.4\columnwidth]{samplesize.eps}
%    \caption{$p$-value of AD test for different sample numbers. The vertical and horizontal black lines indicate the minimum required sample number and the critical value $\alpha$, respectively.}
%    \label{fig:samplesize}
%\end{figure}

%\vspace{-0.2cm}
\subsection{Comparison with Conventional Tail Estimation Method}
\label{sec:PerformanceComparison}
Fig.~\ref{fig:comparison} shows the CDF of the  normalized power for the empirical model, conventional extrapolation-based method, and the proposed Pareto model. 
To apply our proposed methodology, we use Pareto distribution to estimate the upper and lower tails of the channel, denoted by \textit{Pareto Lower Tail} and \textit{Pareto Upper Tail}, respectively, while the kernel-smoothing function in \textsc{MATLAB} is used to estimate the CDF of the interior values between the lower and upper tail probabilities, denoted by \textit{Kernel Smoothed Interior} in the figure.
The optimum threshold for the lower tail has been determined in Section~\ref{sec:performance-declustering} as $-19$ dBm by considering $10^6$ samples. It is worth mentioning that, for the sake of simplicity and without loss of generality, we purposefully ignore threshold modeling for the upper tail as it is not the aim of this study, thus, the same sample portion in the upper tail is used as that of the lower tail. To obtain the extrapolated Weibull curve, we extract the observations corresponding to the reliability order of $10^{-3}-10^{0}$ PER, i.e. the observations whose probability of occurrences are between $1/{10^{3}}$ and $1/{10^{0}}$. Upon applying different distributions to these samples, we select the best-fitting distribution as the Weibull distribution according to the Akaike information criterion/Bayesian information criterion (AIC/BIC) metric, and then compute the extrapolated tail distribution. In addition, we extract first $10^3$ observations regardless of their probability of occurrences, and then apply different distributions to obtain the best fitted one. In contrast to the aforementioned Weibull case, this time, the Rician distribution is found to be the best fitting distribution to the first $10^3$ observations according to the AIC/BIC metric. Then, to obtain the extrapolated Rician curve in Fig.~\ref{fig:comparison}, we compute the corresponding extrapolated tail.
This figure reveals that the traditional extrapolation-based approach strongly depends on the portion of the data used for the modeling, resulting in the over-estimation or under-estimation of the empirical results.
The proposed method has been demonstrated to perform significantly better than the conventional extrapolation-based approach, especially in the ultra-reliability region. The proposed model outperforms the extrapolated Rician and Weibull models by $0.01$ and $1.94 \times 10^{-4}$ in terms of RMSE, respectively, which is a significant improvement for the ultra-reliability region, where the reliability orders are in the range of $10^{-9}-10^{-5}$.

\begin{figure}[h]
   \begin{minipage}{0.48\textwidth}
     \centering
     \includegraphics[width=.7\linewidth]{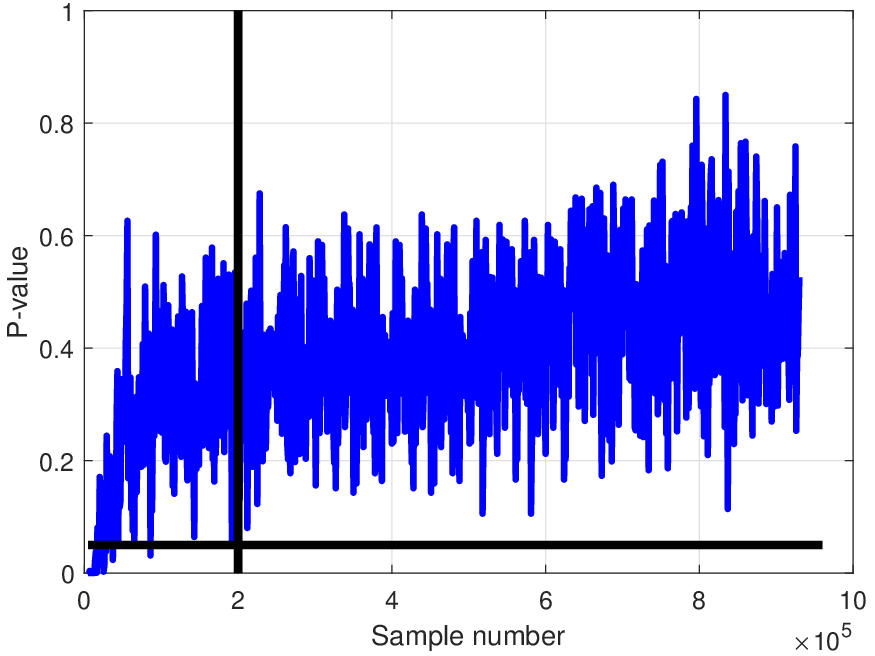}
     %\vspace{-0.2cm}
     \caption{$p$-value of AD test for different sample numbers. The vertical and horizontal black lines indicate the minimum required number of samples and the critical value $\alpha$, respectively.}\label{fig:samplesize}
   \end{minipage}\hfill
   \vspace{0.3cm}
   \begin{minipage}{0.48\textwidth}
     \centering
     \includegraphics[width=.7\linewidth]{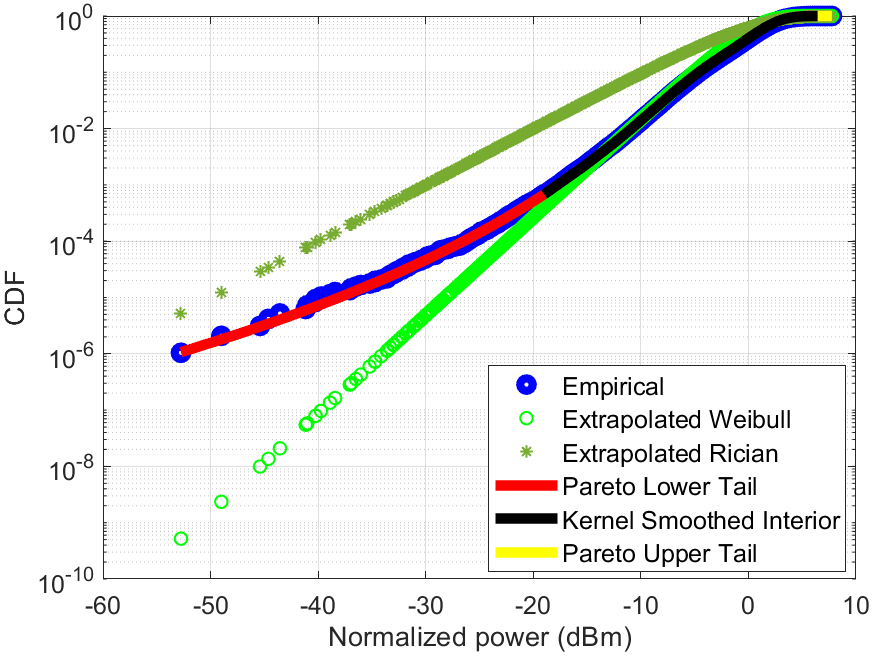}
     \caption{CDF of the normalized power for empirical, extrapolation-based, and Pareto model.}\label{fig:comparison}
   \end{minipage}
\end{figure}
%\begin{figure}[ht]
%    \centering
%    \includegraphics[width=0.4\columnwidth]{comparisonWeibullRician.eps}
%    \caption{CDF of the normalized power for empirical, extrapolation-based, and Pareto model.}
%    \label{fig:comparison}
%\end{figure}

\section{Conclusions}
\label{sec:conclusions}
In this paper, we introduce a novel framework based on the extreme value theory with the goal of estimating the channel tail statistics for URLLC. 
The proposed methodology adopts EVT to determine the optimum threshold over which the tail statistics are derived, the stopping condition to specify the minimum number of samples required to model the tail characteristics, and the assessment to validate the final model by using probability plots. 
The proposed algorithm for determining the optimal stopping condition provides a significantly lower number of required samples for estimating the channel tail distribution. Moreover, the proposed channel modeling methodology achieves  a significantly better fit to the empirical data in the lower tail than the conventional extrapolation-based approach, especially in the ultra-reliability region. Furthermore, the conventional approach has been demonstrated to be highly dependent on the portion of the data considered in the extrapolation, which can seriously affect the reliability performance. 
In the future, we are planning to extend the proposed framework for the EVT analysis of the non-stationary processes to investigate the variation of the statistics of the extreme events over time \cite{mehrnia2021}, and point process approach to calculate the expected waiting time until the next extreme event happens, given the probability of occurrence of the event. Additionally, we are planning to include a more extensive set of parameters that affect the extreme events, such as vehicle speed, road material, and temperature. Also, the extension of this work for assessing the performance of the proposed algorithm at lower outage probabilities as well as deriving the mathematical expressions of BER and PER due to the outage is subject to future work.
\balance 

\ifCLASSOPTIONcaptionsoff
  \newpage
\fi

\bibliographystyle{ieeetr}
\bibliography{EVTChannel.bib}

\begin{IEEEbiography}[{\includegraphics[width=1in,height=1.25in,clip,keepaspectratio]{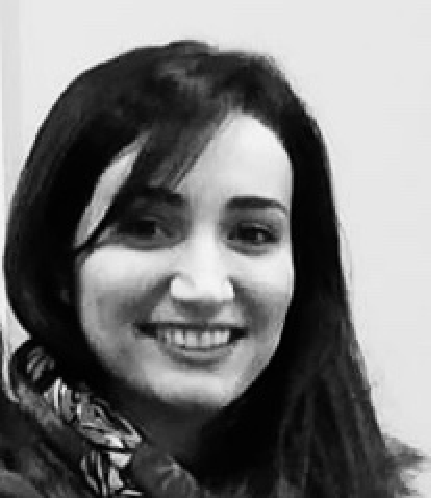}}]%
{Niloofar Mehrnia}
received the B.S. degree in biomedical engineering and M.S. degree in electrical and electronics engineering from Isfahan university and Istanbul Sehir university in 2014 and 2017, respectively. She is currently a Ph.D. student in electrical and electronics engineering at Koc University working in the Wireless Networks Laboratory under the supervision of Prof. Sinem Coleri. Since 2019, she is also working at Koc university Ford Otosan Automotive Technologies Laboratory (KUFOTAL) as a research assistant. Her research interests are wireless channel modeling, ultra-reliable and low-latency communications, and inter-vehicular communications.
\end{IEEEbiography}

%\begin{IEEEbiography}{Niloofar Mehrnia} received the B.S. in biomedical engineering and M.S. in electrical and electronics engineering at Isfahan university and Istanbul Sehir university in 2014 and 2017, respectivelty. She is currently a Ph.D. student in electrical and electronics engineering at Koc University working in the Wireless Networks Laboratory under the supervision of Prof. Sinem Coleri. Since 2019, she is working at Koc university Ford Otosan Automotive Technologies Laboratory (KUFOTAL) as a research assistant. Her research interests are wireless channel modeling, ultra-reliable and low-latency communications, and inter-vehicular communications.
%\end{IEEEbiography}

\begin{IEEEbiography}[{\includegraphics[width=1in,height=1.25in,clip,keepaspectratio]{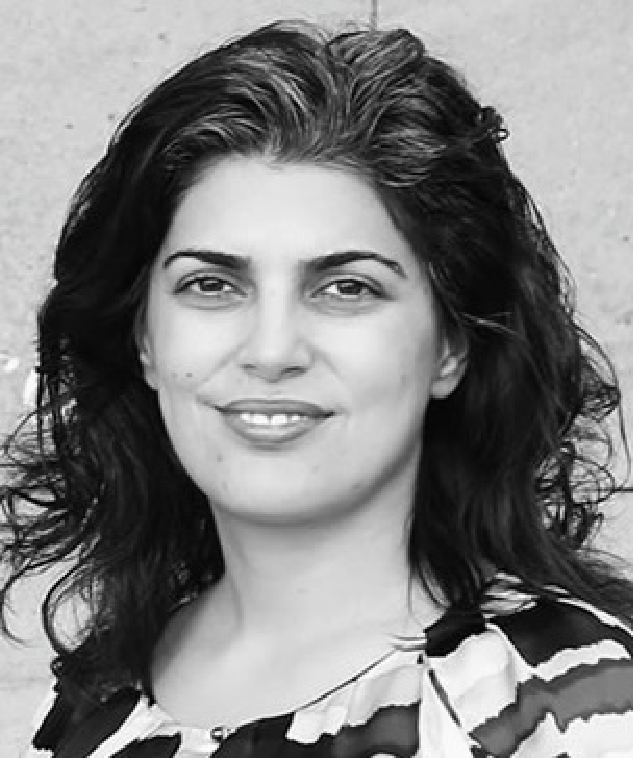}}]{Sinem Coleri}
received the BS degree in electrical and electronics engineering from Bilkent University in 2000, the M.S. and Ph.D. degrees in electrical engineering and computer sciences from University of California Berkeley in 2002 and 2005. She worked as a research scientist in Wireless Sensor Networks Berkeley Lab under sponsorship of Pirelli and Telecom Italia from 2006 to 2009. Since September 2009, she has been a faculty member in the department of Electrical and Electronics Engineering at Koc University, where she is currently Professor. Her research interests are in wireless communications and networking with applications in cyber-physical systems, machine-to-machine communication, sensor networks and intelligent transportation systems. 

She received TUBITAK (The Scientific and Technological Research Council of Turkey) Incentive Award and IEEE Vehicular Technology Society 2020 Neal Shepherd Memorial Best Propagation Paper Award in 2020, College of Engineering Outstanding Faculty Award at Koc University and IEEE Communications Letters Exemplary Editor Award as Area Editor in 2019, Outstanding Achievement Award by Higher Education Council and Academician of the Year Award by ANTIKAD (Antalya Businesswoman Association) in 2018, IEEE Communications Letters Exemplary Editor Award and METU- Prof. Dr. Mustafa Parlar Foundation Research Encouragement Award in 2017, IEEE Communications Letters Exemplary Editor Award and Science Heroes Association - Scientist of the Year Award in 2016, Turkish Academy of Sciences Distinguished Young Scientist (TUBA-GEBIP) and TAF (Turkish Academic Fellowship) Network - Outstanding Young Scientist Awards in 2015, Science Academy Young Scientist (BAGEP) Award in 2014, Turk Telekom Collaborative Research Award in 2011 and 2012, Marie Curie Reintegration Grant in 2010, Regents Fellowship from University of California Berkeley in 2000 and Bilkent University Full Scholarship from Bilkent University in 1995. She has been Area Editor of IEEE Communications Letters and IEEE Open Journal of the Communications Society since 2019, Editor of IEEE Transactions on Communications since 2017 and Editor of IEEE Transactions on Vehicular Technology since 2016. 
\end{IEEEbiography}

%\newpage


\end{document}